\documentclass[]{article}

 	\usepackage[latin1]{inputenc}
 	\usepackage[T1]{fontenc}
	
	\usepackage{amsfonts}
 	\usepackage{amsmath}
 	\usepackage{amssymb}
  	\usepackage{amsthm}
  	\usepackage{color}
 	\usepackage[english]{babel}
	\usepackage{enumerate}
	\usepackage{enumitem}
 	\usepackage{latexsym} 	
  	\usepackage[margin=2.75cm]{geometry}
 	\usepackage{setspace}
 	\usepackage{url}
  \usepackage{authblk}
  \usepackage{hyperref}
  \usepackage{color} 
  \definecolor{darkblue}{rgb}{0.0,0.0,0.3}
  \hypersetup{colorlinks,breaklinks,linkcolor=darkblue,urlcolor=darkblue,anchorcolor=darkblue,citecolor=darkblue}

 	\onehalfspace
	\numberwithin{equation}{section}
	
	\newcommand{\rccs}{$RCCS$ }	

	\newcommand{\dsum}[3]{\displaystyle\sum_{{#1}}^{{#2}}{#3}}
	\newcommand{\set}[3]{\left\{{#3}\right\}_{{#1}}^{{#2}}}
	
	\newcommand{\prob}[1]{p\left({#1}\right)}
	\newcommand{\altprob}[1]{p'\left({#1}\right)}
	\newcommand{\cprob}[2]{p\left({#1}\left|{#2}\right.\right)}

	\newtheorem{definition}{Definition}
	\newtheorem{lemma}{Lemma}
	\newtheorem{proposition}{Proposition}

\title{Do Reichenbachian Common Cause Systems of Arbitrary Finite Size Exist?}

\author{}
\date{}

\begin{document}

\maketitle

\vspace{-48 pt}
\begin{center}
\large
Claudio Mazzola \\
\normalsize{School of Historical and Philosophical Inquiry \\ The University of Queensland \\ Forgan Smith Building (1), St. Lucia, QLD 4072, Australia \\ c.mazzola@uq.edu.au}
\end{center}

\begin{center}
\large
Peter Evans \\
\normalsize{School of Historical and Philosophical Inquiry \\ The University of Queensland \\ Forgan Smith Building (1), St. Lucia, QLD 4072, Australia \\ p.evans@uq.edu.au}
\end{center}

\abstract{\noindent The principle of common cause asserts that positive correlations between causally unrelated events ought to be explained through the action of some shared causal factors. Reichenbachian common cause systems are probabilistic structures aimed at accounting for cases where correlations of the aforesaid sort cannot be explained through the action of a single common cause. The existence of Reichenbachian common cause systems of arbitrary finite size for each pair of non-causally correlated events was allegedly demonstrated by Hofer-Szab\'{o} and R\'{e}dei in 2006. This paper shows that their proof is logically deficient, and we propose an improved proof.}


\section{Introduction}


Two random variables are positively correlated just in case their joint probability is greater than the product of the respective probabilities. Positive correlation does not imply causation, yet purely chancy positive correlations are hard to come about. So, whenever two random variables appear to be positively correlated and yet neither of them causes the other, some shared cause is likely to be at work, which could increase the probability of their joint occurrence, and in terms of which their correlation could consequently be explained.

This, in a nutshell, is the \emph{principle of the common cause}. Everyday examples of the principle abound: a physician receiving several patients who exhibit the same symptoms would likely attribute their condition to the same pathogen; a teacher presented with two almost identical research assignments would presumably infer that they had been copied from the same online source; and the simultaneous shutting down of all electrical appliances in a building is most likely attributed to a failure in their common power supply. Nor are instances of the principle hard to find in the empirical sciences. To wit, commonality of phenotypical traits amongst species is typically explained through some common ancestry; the distribution of iron filings along concentric patterns is explained through the action of a central magnetic field; and even the matching shapes of continents are explained through their detachment from a single original land.

Reichenbach \cite{reichenbach 1956} first gave formal shape to the principle by demanding that instances of non-causal positive correlations be attributed to the presence of a \emph{conjunctive fork}.\footnote{For alternative formalisations of the principle  see for instance \cite{salmon 1984} and \cite{mazzola 2013}. For a comprehensive review of the current status of the principle in probabilistic causal modelling, see \cite{hofer-szabo redei szabo 2013}.} Informally, conjunctive forks are elementary causal structures whereby a common cause increases the probability of two otherwise independent effects. Put formally:
\begin{definition}
Let $(\Omega, p)$ be a classical probability space with $\sigma$-algebra of random events $\Omega$ and probability measure $p$. For any three distinct $A, B, C\in \Omega$, the ordered triple $\langle A,B,C\rangle$ is a \emph{conjunctive fork} for the pair $(A,B)$ if and only if:
\begin{gather}
1\neq \prob{C} \neq 0 													\label{c-fork 0} \\
\cprob{A\wedge B}{C} - \cprob{A}{C}\cprob{B}{C}  =  0					 		  	\label{c-fork 1} \\
\cprob{A\wedge B}{\overline{C}} - \cprob{A}{\overline{C}}\cprob{B}{\overline{C}}  =  0						\label{c-fork 2} \\
\cprob{A}{C} - \cprob{A}{\overline{C}}  >  0											\label{c-fork 3} \\
\cprob{B}{C} - \cprob{B}{\overline{C}}  >  0.										\label{c-fork 4}
\end{gather}
\end{definition}

Conjunctive forks are meant to explain positive correlations in two ways. First, whenever a conjunctive fork of the form $\langle A,B,C\rangle$ exists, the pair $(A,B)$ can be easily demonstrated to be positively correlated:
\begin{proposition}
Let $(\Omega, p)$ be a classical probability space with $\sigma$-algebra of random events $\Omega$ and probability measure $p$. For any three distinct $A, B, C\in \Omega$, if $\langle A, B, C\rangle$ is a conjunctive fork, then:
\begin{equation}\label{equation: positive correlation}
\prob{A\wedge B} - \prob{A}\prob{B} > 0.
\end{equation}
\end{proposition}
Second, conditions (\ref{c-fork 1})--(\ref{c-fork 2}) in the definition of a conjunctive fork declare that the correlation in (\ref{equation: positive correlation}) should vanish conditional on the occurrence of the common cause. This means that such correlation is purely epiphenomenal, being a mere by-product of said cause.

Still, a conjunctive fork may not exist for each pair of causally independent but positively correlated events $(A,B)$ in the given probability space. The principle of the common cause can be preserved in the face of this limitation in two ways. On the one hand, one may look for conjunctive forks in a larger probability space \cite{hofer-szabo redei szabo 1999}; alternatively, one may attribute the correlation not to a single common cause, but rather to the collective action of a plurality of common causes. In \cite{hofer-szabo redei 2004} Hofer-Szab\'{o} and R\'{e}dei pursue the latter strategy, introducing what they called \emph{Reichenbachian common cause systems}:

\begin{definition}\label{definition: rccs}
Let $(\Omega, p)$ be a classical probability space with $\sigma$-algebra of random events $\Omega$ and probability measure $p$. For any $A,B\in \Omega$, a \emph{Reichenbachian common cause system} (RCCS) of size $n\geq 2$ for the pair $A,B$ is a partition $\set{i=1}{n}{C_i}$ of $\Omega$ such that:
\begin{align}
\prob{C_i} \neq 0								 	&	&	\phantom{12 pt}	(i = 1, ..., n)					\label{rccs 0}\\	
\cprob{A\wedge B}{C_i} - \cprob{A}{C_i}\cprob{B}{C_i} = 0               	&	& 	\phantom{12 pt}	(i = 1, ..., n)					\label{rccs 1}\\
[\cprob{A}{C_i}-\cprob{A}{C_j}][\cprob{B}{C_i}-\cprob{B}{C_j}]> 0      &	&	\phantom{12 pt}	(1, ..., n = i\neq j = 1, ..., n).	   	\label{rccs 2}
\end{align}
\end{definition}

$RCCS$s are claimed to generalise the notion of a conjunctive fork in two respects. On the one hand, conditions (\ref{rccs 0}), (\ref{rccs 1}) and (\ref{rccs 2}) are intended to generalise constraints (\ref{c-fork 0}), (\ref{c-fork 1})--(\ref{c-fork 2}) and (\ref{c-fork 3})--(\ref{c-fork 4}) respectively. On the other hand, Hofer-Szab\'{o} and R\'{e}dei show that every \rccs determines a positive correlation between the terms of the corresponding pair, thereby emulating the explanatory function of conjunctive forks.%
\footnote{Whether $RCCS$s are an accurate generalisation of conjunctive forks is controversial. See \cite{mazzola 2012} and \cite{stergiou 2015} for two opposing views on this matter.} %
But do $RCCS$s actually exist for each positively correlated pair of causally independent events?

Hofer-Szab\'{o} and R\'{e}dei claim that they do \cite{hofer-szabo redei 2006}. More precisely, they offer a proof to the effect that, for each non-strictly positively correlated pair of events $(A, B)$ in a classical probability space $(\Omega, p)$ and any integer number $n\geq 2$, some extension of $(\Omega, p)$ can be constructed whereby a \rccs of size $n$ can be found for said pair. Their proof, however, is invalid. The aim of this article is to show why it is, and to replace it with an amended one. Each of these goals will be pursued respectively in the following two sections.


\section{Diagnosis}


Hofer-Szab\'{o} and R\'{e}dei's proof is articulated into two major steps, which can be briefly summarised as follows:

\emph{Step 1}: Hofer-Szab\'{o} and R\'{e}dei notice that conditions (\ref{rccs 0})--(\ref{rccs 2}) and the axioms of classical probability theory jointly constrain the possible values that can be assigned to probabilities $\cprob{A}{C_{1}}$, ...., $\cprob{A}{C_{n}}$, $\cprob{B}{C_{1}}$, ...., $\cprob{B}{C_{n}}$, $\prob{C_{1}}$, ...., $\prob{C_{n}}$,  so that $\set{i=1}{n}{C_i}$ could be a \rccs of size $n$ for the pair $(A, B)$ in a probability space $(\Omega, p)$. They thereby demonstrate that some set of $3n$ real numbers satisfying such constraints can always be found, for each possible value of $n\geq 2$ and any correlated pair $(A, B)$.

\emph{Step 2}: On this basis, Hofer-Szab\'{o} and R\'{e}dei  show how, for each such $n$ and $(A, B)$, an extension ($\Omega', p'$) of $(\Omega, p)$ can be constructed so that $\Omega$ includes a partition $\set{i=1}{n}{C_i}$ and so that the values of $\cprob{A}{C_{1}}$, ...., $\cprob{A}{C_{n}}$, $\cprob{B}{C_{1}}$, ...., $\cprob{B}{C_{n}}$, $\prob{C_{1}}$, ...., $\prob{C_{n}}$ meet the constraints mentioned above. This, they conclude, suffices to prove that some extension of the assumed probability space can in general be found, which includes a \rccs of size $n$ for a given pair.

Step 1 is where the mistake with Hofer-Szab\'{o} and R\'{e}dei's demonstration lies, so we shall specifically focus on it. To facilitate our analysis, it will be convenient to further subdivide Step 1 into three sub-components.

\emph{Step 1a}: The proof begins by stating the set of conditions that a set $\set{n = 1}{n}{a_{i}, b_{i}, c_{i}}$ of $3n$ real numbers must satisfy so that $\set{i=1}{n}{C_i}$ be a \rccs of size $n$ for the pair $(A, B)$ in probability space $(\Omega, p)$. These conditions are:
\begin{gather}
\prob{A} =\dsum{i = 1}{n}{a_{i}c_{i}}																																							\label{equation: adm a}		\\
\prob{B}=\dsum{i = 1}{n}{b_{i}c_{i}}																																							\label{equation: adm b}		\\
\prob{A\wedge B} = \dsum{i = 1}{n}{a_{i}b_{i}c_{i}}																																				\label{equation: adm ab}		\\
1 = \dsum{i=1}{n}{c_{i}}		 																																							\label{equation: adm partition}	\\
0< [a_{i}-a_{j}][b_{i}-b_{j}]    																																\phantom{12 pt}	(1, ..., n = i\neq j = 1, ..., n)	\label{equation: adm rccs 2}	\\	
0 \leq a_{i}, b_{i} \leq 1  																																	\phantom{12 pt}	(i = 1, ..., n)				\label{equation: adm aibi}		\\
0 < c_{i} < 1 																																		\phantom{12 pt}	(i = 1, ..., n),			\label{equation: adm ci}	 	
\end{gather}
under identifications
\begin{gather}
\prob{C_{i}} = c_{i} 									\phantom{12 pt}	(i = 1, ..., n)					\label{equation: adm 1} \\
\cprob{A}{C_{i}} = a_{i} 									\phantom{12 pt}	(i = 1, ..., n)					\label{equation: adm 2} \\
\cprob{B}{C_{i}} = b_{i} 									\phantom{12 pt}	(i = 1, ..., n).				\label{equation: adm 3}
\end{gather}
Hofer-Szab\'{o} and R\'{e}dei call a set of $3n$ numbers satisfying the above constraints \emph{admissible} for the pair $(A,B)$.

\emph{Step 1b}: Next, Hofer-Szab\'{o} and R\'{e}dei point out that equalities (\ref{equation: adm a})--(\ref{equation: adm partition}) effectively restrict the number of independent parameters in each admissible set $\set{n = 1}{n}{a_{i}, b_{i}, c_{i}}$ to $3n- 4$. Specifically, they maintain, numbers $c_{n-1}$, $c_{n}$, $a_{n}$, and $b_{n}$ can be expressed as functions of the remaining  $3n- 4$ parameters as follows:
\setlength{\jot}{12  pt}
\begin{gather}
c_{n-1} = \frac{\gamma + \frac{1}{2}\dsum{j,k=1}{n-2}{c_{j}c_{k}[a_{j}-a_{k}][b_{j}-b_{k}]}- \dsum{k = 1}{n-2}{c_{k}[a-a_{k}][b-b_{k}]}+\gamma\dsum{k = 1}{n-2}{c_{k}}}{[a-a_{n-1}][b-b_{n-1}]+\gamma + \dsum{k= 1}{n-2}{[a_{n-1}-a_{k}][b_{n-1}-b_{k}]}} 				\label{equation: cn-1}	\\
c_{n} = 1-\dsum{i = 1}{n-1}{c_{i}}																																							\label{equation: cn}	\\
a_{n} = \frac{a+\dsum{k = 1}{n-2}{c_{k}[a_{n-1}-a_{k}]}-a_{n-1}}{1-\dsum{k = 1}{n-1}{c_{k}}}+a_{n-1}																													\label{equation: an}	\\
b_{n} = \frac{b+\dsum{k = 1}{n-2}{c_{k}[b_{n-1}-b_{k}]}-b_{n-1}}{1-\dsum{k = 1}{n-1}{c_{k}}}+b_{n-1}																													\label{equation: bn}	
\end{gather}
where, for notational convenience, we set:
\setlength{\jot}{6  pt}
\begin{gather}
a= \prob{A} \\
b= \prob{B} \\
\gamma= \prob{A\wedge B} - \prob{A}\prob{B}.
\end{gather}

\emph{Step 1c}: Finally, thanks to the above equivalences, it is demonstrated by induction on $n$ that a set $\set{i = 1}{n}{a_{i}, b_{i}, c_{i}}$ of admissible numbers for $(A, B)$ exists for each $n\geq 2$. More specifically it is argued that, for any such $n$, some set of $3n - 4$ parameters $\set{i = 1}{n-2}{a_{i}, a_{n-1}, b_{i}, b_{n-1}, c_{i}}$ can always be chosen so that, in virtue of (\ref{equation: adm rccs 2})--(\ref{equation: adm ci}) and (\ref{equation: cn-1})--(\ref{equation: bn}), the resulting set $\set{i = 1}{n}{a_{i}, b_{i}, c_{i}}$ be admissible for $(A, B)$.

Leaving aside some minor and relatively harmless mathematical glitches, the main problem with this part of the proof is logical. Hofer-Szab\'{o} and R\'{e}dei's argument in fact rests on the presupposition that \emph{if} numbers $\set{n = 1}{n}{a_{i}, b_{i}, c_{i}}$ meet the conditions (\ref{equation: adm a})--(\ref{equation: adm ci}) defining an admissible set for $(A, B)$, \emph{then} a partition satisfying identifications (\ref{equation: adm 1})--(\ref{equation: adm 3}) will thereby qualify as a $\rccs$ for the same pair (i.e., will satisfy (\ref{rccs 0})--(\ref{rccs 2})): quite evidently, if that were not the case, demonstrating the existence of admissible numbers could not prove anything about the existence of $RCCS$s. Hofer-Szab\'{o} and R\'{e}dei even make this assumption explicit when declaring that `admissible numbers have been chosen precisely so that [(\ref{rccs 0})--(\ref{rccs 2})] are satisfied' \cite[p. 755]{hofer-szabo redei 2006}. 
Put in other words, they submit that given (\ref{equation: adm 1})--(\ref{equation: adm 3}), conditions (\ref{equation: adm a})--(\ref{equation: adm ci}) jointly imply (\ref{rccs 0})--(\ref{rccs 2}). Nevertheless, and here lies the logical error with their proof, this is demonstrably false.

To verify this, let us hereafter grant (\ref{equation: adm 1})--(\ref{equation: adm 3}). This stipulation immediately turns (\ref{equation: adm aibi}) into a direct consequence  of the axioms of probability theory, whereas (\ref{equation: adm partition}) essentially demands that $\set{i = 1}{n}{C_{i}}$ be a partition of the assumed probability space, as indeed it is required from a $RCCS$. Having established this, (\ref{equation: adm a}) and (\ref{equation: adm b}) are then easily obtained from the theorem of total probability. This shows that all of the aforesaid conditions are actually independent from (\ref{rccs 0})--(\ref{rccs 2}), so the real job must be done  by the remaining three equalities. Now (\ref{equation: adm ci}) does effectively imply (\ref{rccs 0}), while (\ref{equation: adm rccs 2}) is openly equivalent to (\ref{rccs 2}). By exclusion, it follows that (\ref{equation: adm a})--(\ref{equation: adm ci}) jointly imply (\ref{rccs 0})--(\ref{rccs 2}), as Hofer-Szab\'{o} and R\'{e}dei's proof requires, if and only if (\ref{equation: adm ab}) logically implies (\ref{rccs 1}). However, this is not so.

To illustrate, let $X, Y\in \Omega$ be arbitrarily chosen from $(\Omega, p)$, and let $\set{i = 1}{n}{Z_{i}}\subseteq \Omega$ be a partition of that space. The theorem of total probability then produces the following, general formula:
\begin{equation}\label{equation: ab general}
\prob{X\wedge Y} = \dsum{i = i}{n}{\cprob{X}{Z_{i}}\cprob{Y}{Z_{i}}\prob{Z_i}} + \dsum{i = i}{n}{\prob{Z_{i}}[\cprob{X\wedge Y}{Z_{i}}-\cprob{X}{Z_{i}}\cprob{Y}{Z_{i}}]}.
\end{equation}
Now let $A = X$, $B = Y$ and $Z_{i} = C_{i}$ for $i = 1, ..., n$. To get (\ref{equation: adm ab}) from (\ref{equation: ab general}), it is then necessary and sufficient that:
\begin{equation}\label{equation: crirical condition for ab}
\dsum{i = 1}{n}{\prob{Z_{i}}[\cprob{X\wedge Y}{Z_{i}}-\cprob{X}{Z_{i}}\cprob{Y}{Z_{i}}]} = 0.
\end{equation}
Now, if (\ref{rccs 1}) is true, then surely (\ref{equation: crirical condition for ab}) is the case and (\ref{equation: adm ab}) follows as a result. Therefore, (\ref{rccs 1}) is clearly a sufficient condition for (\ref{equation: adm ab}), while conversely (\ref{equation: adm ab}) is certainly a necessary condition for (\ref{rccs 1}). Still, it is equally clear that  (\ref{rccs 1}) being the case is not the \emph{sole} circumstance whereby (\ref{equation: adm ab}) can be so obtained. Because each term of the form $\cprob{X\wedge Y}{Z_{i}}-\cprob{X}{Z_{i}}\cprob{Y}{Z_{i}}$ can perfectly well be positive or negative,  (\ref{equation: crirical condition for ab}) can in fact be satisfied even  in case not all such terms are equal to zero. This is easily checked by putting, for instance,
\begin{eqnarray*}
\prob{C_{1}} =  \frac{1}{2} & &  \prob{C_{2}} = \frac{1}{2}\\
\cprob{A}{C_{1}} = \frac{1}{4} & & \cprob{A}{C_{2}} = \frac{1}{8} \\
\cprob{B}{C_{1}} = \frac{1}{3} & & \cprob{B}{C_{2}} = \frac{1}{6}\\
\cprob{A\wedge B}{C_{1}} = \frac{1}{24} & & \cprob{A\wedge B}{C_{2}} = \frac{1}{16}
\end{eqnarray*}
where $n = 2$; (\ref{equation: crirical condition for ab}) is satisfied but (\ref{rccs 1}) is violated.

This proves that (\ref{rccs 1}) is not a necessary condition for (\ref{equation: adm ab}), so the latter does not logically imply the former. Contrary to what Hofer-Szab\'{o} and R\'{e}dei declare, admissible numbers have \emph{not} been chosen so that (\ref{rccs 0})--(\ref{rccs 2}) are satisfied; so proving the existence of admissible numbers cannot in any way contribute to demonstrating the existence of $RCCS$s.  Hofer-Szab\'{o} and R\'{e}dei's proof is consequently invalid.

Before proceeding further, it will be worth stopping briefly to consider a possible rejoinder. Looking back at Hofer-Szab\'{o} and R\'{e}dei's proof, a chain of two subsequent inferences can be seen at work. First, (\ref{equation: adm a})--(\ref{equation: adm partition}) are inferred from (\ref{equation: adm rccs 2})--(\ref{equation: adm ci}), (\ref{equation: cn-1})--(\ref{equation: bn}); and subsequently, (\ref{rccs 0})--(\ref{rccs 2}) are allegedly inferred from (\ref{equation: adm a})--(\ref{equation: adm ci}).\footnote{Notice that the logical direction of the proof is opposite to the direction in which it is presented.} We have shown that the latter inferences in the chain are invalid, but what if (\ref{rccs 0})--(\ref{rccs 2}) could be directly obtained from (\ref{equation: adm rccs 2})--(\ref{equation: adm ci}), without the mediation of  (\ref{equation: adm a})--(\ref{equation: adm partition})? If that were the case, then the crux of Hofer-Szab\'{o} and R\'{e}dei's proof could be salvaged after all.

A few elementary calculations would show that, in effect, (\ref{equation: cn}), (\ref{equation: an}) and (\ref{equation: bn}) are respectively equivalent to (\ref{equation: adm partition}), (\ref{equation: adm a}) and (\ref{equation: adm b}). Our problem therefore reduces to establishing whether (\ref{equation: cn-1}) does really imply (\ref{rccs 1}). On the face of it, this is not quite clear. Hofer-Szab\'{o} and R\'{e}dei seem to treat (\ref{equation: cn-1}) as equivalent to (\ref{equation: adm ab}), but this might be only because they think, erroneously, that the latter be equivalent to (\ref{rccs 1}). Indeed, they do not directly obtain (\ref{equation: cn-1})  from  (\ref{equation: adm ab}), but they derive it instead from this other formula:
\begin{equation}\label{equation: correlation special}
\prob{X\wedge Y} - \prob{X}\prob{Y} =  \frac{1}{2}\dsum{i,j = 1}{n}{\prob{Z_{i}}\prob{Z_{j}}[\cprob{X}{Z_{i}}-\cprob{X}{Z_{j}}][\cprob{Y}{Z_{i}}-\cprob{Y}{Z_{j}}]}.
\end{equation}

This equation can be easily demonstrated to hold for every pair of random events $X$ and $Y$ in a classical probability space and every partition $\set{n = 1}{i}{Z_{i}}$ of that space for which (\ref{rccs 1}) is true. Thus, if Hofer-Szab\'{o} and R\'{e}dei's calculations are correct, (\ref{rccs 1}) is at least a sufficient condition for (\ref{equation: cn-1}). But does (\ref{equation: correlation special}) show that (\ref{equation: cn-1}) in turn implies (\ref{rccs 1})?

To answer this question we need to look at an even further formula, namely:
\begin{eqnarray}
\prob{X\wedge Y} - \prob{X}\prob{Y} 	&   =  	&	\frac{1}{2} \dsum{i,j = 1}{n}{\prob{Z_{i}}\prob{Z_{j}}[\cprob{X}{Z_{i}}-\cprob{X}{Z_{j}}][\cprob{Y}{Z_{i}}-\cprob{Y}{Z_{j}}]} + \nonumber\\
						&	& 	\frac{1}{2} \dsum{i = 1}{n}{\prob{Z_{i}}[\cprob{X\wedge Y}{Z_{i}}-\cprob{X}{C_{i}}\cprob{Y}{C_{i}}]}. 								 \label{equation: correlation general}
\end{eqnarray}
This equality follows from the theorem of total probability alone, and it is true of all $X$ and $Y$ and all partitions $\set{n = 1}{i}{Z_{i}}$ of a classical probability space. Formula (\ref{equation: correlation special}) is clearly a special case of this more general equality, and it obtains, once again, if and only if:
\begin{equation}\label{equation: crirical condition for correlation}
\dsum{i = 1}{n}{\prob{Z_{i}}[\cprob{X\wedge Y}{Z_{i}}-\cprob{X}{Z_{i}}\cprob{Y}{Z_{i}}]} = 0.
\tag{\ref{equation: crirical condition for ab}}
\end{equation}
On the other hand, it should be clear by now that (\ref{rccs 1}) is merely a \emph{sufficient} condition for (\ref{equation: crirical condition for ab}). Put in other words, it is possible that (\ref{equation: crirical condition for ab}), and hence (\ref{equation: correlation special}), be satisfied without (\ref{rccs 1}) being true. Consequently, it is possible that (\ref{equation: cn-1}) be satisfied whilst (\ref{rccs 1}) is not, so the rejoinder presently under examination would be unsuccessful. Unfortunately, Hofer-Szab\'{o} and R\'{e}dei's proof cannot be salvaged in this way.


\section{Revision}


Our proof, as it will become apparent, will closely follow the structure of  Hofer-Szab\'{o} and R\'{e}dei's one, and all differences will ultimately depend on the different form of the mathematical constraints involved. For this reason, it will be helpful to articulate it into four sub-steps, each one corresponding to one step from the original proof.

\emph{Step 1a*}: The error with  Hofer-Szab\'{o} and R\'{e}dei's proof, as we saw, rests in their definition of admissible numbers. More specifically, it lies in the fact that (\ref{equation: adm ab}) fails to fully capture condition (\ref{rccs 1}), for which it is a necessary but not sufficient condition. The first step in rectifying Hofer-Szab\'{o} and R\'{e}dei's proof, therefore, will consist in modifying their characterization of admissible numbers. Let us accordingly lay down the following definition:

\begin{definition}\label{definition: admissible*}
Let $(\Omega, p)$ be a classical probability space with $\sigma$-algebra of random events $\Omega$ and probability measure $p$. For any $A,B\in \Omega$ satisfying (\ref{equation: positive correlation}) and any $n \geq 2$, the set
\begin{equation*}
\set{i = 1}{n}{a_{i}, b_{i}, c_{i}, d_{i}}
\end{equation*}
 is called \emph{admissible* for} $(A,B)$  if and only if the following conditions hold:
\begin{gather}
\prob{A} =\dsum{i = 1}{n}{a_{i}c_{i}}																																							\label{equation: adm a*}		\\
\prob{B}=\dsum{i = 1}{n}{b_{i}c_{i}}																																							\label{equation: adm b*}		\\
1 = \dsum{i =1}{n}{c_{i}}																																									\label{equation: adm partition*}	\\
0= d_{i}- a_{i}b_{i} 																																	\phantom{12 pt}	(i = 1, ..., n)				\label{equation: adm di*}		\\
0< [a_{i}-a_{j}][b_{i}-b_{j}]    																																\phantom{12 pt}	(1, ..., n = i\neq j = 1, ..., n)	\label{equation: adm rccs 2*}	\\	
0 <a_{i}, b_{i}, d_{i} < 1  																																\phantom{12 pt}	(i = 1, ..., n)				\label{equation: adm aibi*}	\\
0 < c_{i} < 1 																																		\phantom{12 pt}	(i = 1, ..., n).			\label{equation: adm ci*}	 	
\end{gather}
\end{definition}

Notice the above definition is largely similar to the one informally proposed by Hofer-Szab\'{o} and R\'{e}dei, save for the fact that (\ref{equation: adm ab}) is replaced by the stronger condition (\ref{equation: adm di*}), and for the fact that admissible*  sets contain $n$ more numbers $\set{i = 1}{n}{d_{i}}$ than  admissible sets, for each finite value of $n$. The latter change was indeed made necessary in order to include (\ref{equation: adm di*}).

Before we proceed any further, we ought to be sure that the above definition does not suffer from the same shortcomings as the one that it is intended to replace. This is established through the following lemma \--- proof of which is elementary, and which has been consequently omitted:
\begin{lemma}\label{lemma: admissible* numbers}
Let $(\Omega, p)$ be a classical probability space with $\sigma$-algebra of random events $\Omega$ and probability measure $p$. For any $A,B\in \Omega$ satisfying (\ref{equation: positive correlation}) and any $\set{i = 1}{n}{C_{i}}\subseteq\Omega$ where $n\geq 2$, the set $\set{i = 1}{n}{C_{i}}$ is  a \rccs  of size $n$ for $(A, B)$ if and only if there exists a set $\set{i = 1}{n}{a_{i}, b_{i}, c_{i}, d_{i}}$ of admissible* numbers for $(A, B)$such that:
\begin{gather}
\prob{C_{i}} = c_{i} 									\phantom{12 pt}	(i = 1, ..., n)					\label{equation: adm* 1} \\
\cprob{A}{C_{i}} = a_{i} 									\phantom{12 pt}	(i = 1, ..., n)					\label{equation: adm* 2} \\
\cprob{B}{C_{i}} = b_{i} 									\phantom{12 pt}	(i = 1, ..., n)					\label{equation: adm* 3} \\
\cprob{A\wedge B}{C_{i}} = d_{i}							\phantom{12 pt}	(i = 1, ..., n).				\label{equation: adm* 4}
\end{gather}
\end{lemma}

\emph{Step 1b*}: Now that Definition \ref{definition: admissible*} is securely in place, we can continue with our demonstration. Following Hofer-Szab\'{o} and R\'{e}dei, our next step will be to restrict the number of independent parameters needed to identify an admissible* set. This will be achieved by means of our second  lemma:
\begin{lemma}\label{lemma: admissible* variant}
Let $(\Omega, p)$ be a classical probability space with $\sigma$-algebra of random events $\Omega$ and probability measure $p$. Let $A,B\in \Omega$ satisfy (\ref{equation: positive correlation}), let set $\set{i = 1}{n}{C_{i}}\subseteq\Omega$ with $n\geq 2$, and let $\set{i = 1}{n}{a_{i}, b_{i}, c_{i}, d_{i}}$ be defined so as to meet conditions (\ref{equation: adm* 1})--(\ref{equation: adm* 4}). Finally, let (\ref{equation: adm rccs 2*})--(\ref{equation: adm ci*}) obtain. Then, $\set{i = 1}{n}{a_{i}, b_{i}, c_{i}, d_{i}}$ is admissible*  for $(A, B)$ if and only if :
\begin{gather}
a_{n} = \frac{a-\dsum{k=1}{n-1}{c_{k} a_{k}}}{1-\dsum{k=1}{n-1}{c_{k}}}																																	\label{equation: an*}		\\
b_{n} = \frac{b-\dsum{k=1}{n-1}{c_{k} b_{k}}}{1-\dsum{k=1}{n-1}{c_{k}}}																																	\label{equation: bn*}		\\
c_{n} = 1- \dsum{k=1}{n-1}{c_{k}} 																																							\label{equation: cn*}		\\
d_{n} = \frac{\left[a - \dsum{k =1}{n-1}{a_{k}c_{k}}\right]\left[b - \dsum{k = 1}{n-1}{b_{k}c_{k}}\right]}{\left[ 1 - \dsum{k = 1}{n-1}{c_{k}} \right ]^{2}} 																					\label{equation: dn*}		\\
0= d_{k}- a_{k}b_{k} 																																	\phantom{12 pt}	(k = 1, ..., n-1)			\label{equation: di*}		
\end{gather}
where
\setlength{\jot}{6  pt}
\begin{gather}
a := \prob{A}																		\label{equation: a}			\\
b := \prob{B}.																		\label{equation: b}			
\end{gather}
\end{lemma}

\begin{proof}
Let $(\Omega, p)$ be a classical probability space with $\sigma$-algebra of random events $\Omega$ and probability measure $p$. Moreover,  let $A,B\in \Omega$ be positively correlated and let the set $\set{i = 1}{n}{a_{i}, b_{i}, c_{i}, d_{i}}$ of $n\geq 2$ real numbers satisfy conditions  (\ref{equation: adm* 1})--(\ref{equation: adm* 4}) and (\ref{equation: adm rccs 2*})--(\ref{equation: adm ci*}). Finally, let identities (\ref{equation: a}) and (\ref{equation: b}) be in place.

To begin with let us  observe that, owing to the assumptions just stated and owing to the theorem of total probability, (\ref{equation: adm a*})--(\ref{equation: adm partition*}) are in fact equivalent to (\ref{equation: an*})--(\ref{equation: cn*}). Our proof  thus reduces to showing that (\ref{equation: adm di*}) obtains if and only if (\ref{equation: dn*})--(\ref{equation: di*}) do. This will be done with the help of the following formula
\begin{equation}\label{equation: dn general}
d_{n} = [d_{n} - a_{n}b_{n}] + \frac{\left[a - \dsum{k = i}{n-1}{a_{k}c_{k}}\right]\left[b - \dsum{k = 1}{n-1}{b - b_{k}}\right]}{\left[ 1 - \dsum{k = 1}{n-1}{c_{k}}\right ]^{2}}, 																									
\end{equation}
which, it must be noticed, can be obtained from the theorem of total probability alone. This is to say that (\ref{equation: dn general}) holds quite generally, and that it does not depend on any specific assumptions other than the obvious request that
\begin{gather}
\left[ 1 - \dsum{k = 1}{n-1}{c_{k}}\right ]^{2} \neq 0.
\end{gather}
Clearly, in our case this is guaranteed by (\ref{equation: adm ci*}) and (\ref{equation: adm partition*}).

Thanks to the above equality, proving necessity becomes elementary. If (\ref{equation: adm di*})  holds, then (\ref{equation: di*}) follows a fortiori, whereas to obtain (\ref{equation: dn*}) we only need to observe that (\ref{equation: adm di*}) makes the first addendum in (\ref{equation: dn general}) equal to zero. Let us now turn to sufficiency, so let (\ref{equation: dn*}) and (\ref{equation: di*}) be in place. Once again, the first addendum in (\ref{equation: dn general}) becomes equal to zero, this time thanks to (\ref{equation: dn*}). This result, along with (\ref{equation: di*}), suffices to establish (\ref{equation: adm di*}).
\end{proof}

\emph{Step 1c*}: Just like conditions (\ref{equation: cn-1})--(\ref{equation: bn}) did for admissible sets in the original proof, equations (\ref{equation: an*})--(\ref{equation: di*})  effectively reduce the number of independent parameters whose existence is to be proved in order to establish the existence of admissible* sets of arbitrary finite cardinality. Let us observe, in particular, that (\ref{equation: di*}) is in reality a system of $n-1$ equations, namely one for each value of $i = 1, ..., n-1$. Therefore, (\ref{equation: an*})--(\ref{equation: di*}) jointly comprise $4 + (n - 1) = n + 3$ equations in $4 n$ variables. This means that, for each $n\geq 2$, each admissible* set for $(A, B)$ is determined by a set of $4n - (n+3) = 3n- 3$ parameters, whose existence we are now going to establish by induction.

\begin{lemma}\label{lemma: admissible* numbers exist}
Let $(\Omega, p)$ be a classical probability space with $\sigma$-algebra of random events $\Omega$ and probability measure $p$. For any $A,B\in \Omega$ satisfying (\ref{equation: positive correlation}), a set $\set{i = 1}{n}{a_{i}, b_{i}, c_{i}, d_{i}}$ of admissible* numbers for $(A, B)$exists for each $n\geq 2$.
\end{lemma}

\begin{proof}
Let $(\Omega, p)$ be a classical probability space with $\sigma$-algebra of random events $\Omega$ and probability measure $p$. Moreover, let $A,B\in \Omega$ satisfy (\ref{equation: positive correlation}), and let us further assume, for ease of expression, identities (\ref{equation: a})--(\ref{equation: b}). Proof will proceed by induction on $n$.

First, let $n = 2$ be our inductive basis: in this case, equations (\ref{equation: an*})--(\ref{equation: di*}) reduce to
\begin{gather}
a_{2} = \frac{a - c_{1} a_{1}}{1-c_{1}}																																			\label{equation: a2}		\\
b_{2} = \frac{b- c_{1}  b_{1}}{1-c_{1}}																																			\label{equation: b2}		\\		
c_{2} = 1- c_{1}																																							\label{equation: c2}		\\
d_{2} = \frac{[a - a_{1}c_{1}][b - b_{1}c_{1}]}{[ 1 - c_{1}]^{2}} 																															\label{equation: d2}		\\
d_{1} =  a_{1}b_{1}. 																																						\label{equation: d1}		     												
\end{gather}	
Because $a$ and $b$ and are given, choosing numbers $c_{1}$, $a_{1}$ and $b_{1}$ will therefore suffice to fix the values of all $4n = 8$ variables in the system. Now let us observe that, as a direct consequence of (\ref{equation: positive correlation}), numbers $a$ and $b$ must lie strictly between one and zero:
\begin{gather}
1 > a,b > 0.
\end{gather}
Since we are interested simply in the existence of some admissible* set, we are free to assume that
\begin{gather}
 c_{1} \rightarrow  0 \\
1> a > a_{1} >  0 \\
1> b > b_{1} > 0
\end{gather}
in order to ensure that
\begin{gather}
0 <a_{1}, a_{2}, b_{1}, b_{2}, d_{1}, d_{2} < 1  																																		\label{equation: a1a2b1b2}	\\
0 < c_{1}, c_{2} < 1 																																						\label{equation: c1c2}	 	\\
0< [a_{1}-a_{2}][b_{1}-b_{2}],																																					\label{equation: a1b1a2b2 adm rccs 2}	
\end{gather}	
which is straightforward to check. The set $\set{i = 1}{2}{a_{i}, b_{i}, c_{i}, d_{i}}$ so defined accordingly satisfies  (\ref{equation: adm rccs 2*})--(\ref{equation: adm ci*}),  in addition to (\ref{equation: an*})--(\ref{equation: di*}). By Lemma \ref{lemma: admissible* variant}, it is therefore admissible* for $(A, B)$.

Next, let $n = m>2$ and let $\set{i = 1}{m}{a_{i}, b_{i}, c_{i}, d_{i}}$ be admissible* for $(A, B)$ as our inductive hypothesis. To prove that an admissible* set for $(A, B)$ exists if $n = m +1$, let us consider the following subset of the admissible* set for $n=m$
\begin{equation*}
\set{j = 1}{m-1}{a_{j}, b_{j}, c_{j}, d_{j}}\subset\set{i = 1}{m}{a_{i}, b_{i}, c_{i}, d_{i}},
\end{equation*}
which we take by assumption to constitute a set of $3(m-1)$ parameters. Let us then choose numbers $a'_{m}, b'_{m}, c'_{m}$ such that
\begin{gather}
a_{j} > a'_{m}  > 0 			\phantom{12 pt}	( j = 1, ..., m-1)	\label{equation: a'm} \\
b_{j} > b'_{m}  > 0 			\phantom{12 pt}	( j = 1, ..., m-1)	\label{equation: b'm} \\
c_{j} > c'_{m} > 0 				\phantom{12 pt}	( j = 1, ..., m-1).	\label{equation: c'm}
\end{gather}
Given (\ref{equation: an*})--(\ref{equation: di*}), the set
\begin{gather*}
\set{j = 1}{m-1}{a_{j}, b_{j}, c_{j}, d_{j}}\cup\set{}{}{a'_{m}, b'_{m}, c'_{m}}
\end{gather*}
of $3(m-1)+3 = 3(m+1)-3 = 3n-3$ parameters will then suffice to determine $4(m+1)$ numbers:
\begin{gather*}
\set{j = 1}{m-1}{a_{j}, b_{j}, c_{j}, d_{j}, a'_{m}, b'_{m}, c'_{m}, d'_{m}, a_{m+1}, b_{m+1}, c_{m+1}, d_{m+1}}.
\end{gather*}
Because conditions (\ref{equation: adm rccs 2*})--(\ref{equation: adm ci*}) hold by assumption for numbers $\set{j = 1}{m-1}{a_{j}, b_{j}, c_{j}, d_{i}, a'_{m}, b'_{m}, c'_{m},}$, all we need to show is that said constraints be also satisfied by $\set{}{}{a_{m+1}, b_{m+1}, c_{m+1}, d'_{m}, d_{m+1}}$. To this purpose, let us first notice that  (\ref{equation: adm aibi*}) must be true of $d_{m}$ by virtue of (\ref{equation: di*}) and (\ref{equation: a'm})--(\ref{equation: b'm}). Next, thanks to (\ref{equation: an*})--(\ref{equation: dn*}), it will be sufficient to suppose that
\begin{gather}
c'_{m} \rightarrow 0
\end{gather}
to obtain
\begin{gather}
a_{m+1} \rightarrow a_{m}\\
b_{m+1} \rightarrow b_{m}\\
c_{m+1} \rightarrow c_{m}\\
d_{m+1} \rightarrow d_{m},
\end{gather}
which we already know, by our inductive hypothesis, to satisfy  (\ref{equation: adm rccs 2*})--(\ref{equation: adm ci*}).  Due once again to Lemma \ref{lemma: admissible* variant}, the set of $4(m+1)$ numbers so determined is therefore admissible* for $(A, B)$. This completes our inductive proof.
\end{proof}

\emph{Step 2*}: Now that we are in possession of $4n$ admissible* numbers for each arbitrary integer value of $n\geq 2$ and every correlated pair $(A, B)$, the remainder of our proof will simply mimic the logical machinery set up by Hofer-Szab\'{o} and R\'{e}dei in the latter part of their proof. That is, admissible* numbers will hereafter accomplish the task which admissible numbers were supposed to perform in the original proof. To establish the existence of $\rccs$s of arbitrary finite size, we accordingly need only one more definition:

\begin{definition}\label{definition: extension}
Let $(\Omega, p)$ and $(\Omega', p')$ be  classical probability spaces with $\sigma$-algebras of random events $\Omega$ and $\Omega'$ and with probability measures $p$ and $p'$, respectively. Then $(\Omega', p')$ is called an \emph{extension} of $(\Omega, p)$ if and only if there exists an injective lattice homomorphism $h: \Omega \rightarrow \Omega'$, preserving complementation, such that
\begin{equation}
\altprob{h(X)} = \prob{X}	\phantom{12 pt} \mbox{for all } X\in \Omega.
\end{equation}
\end{definition}

Thanks to this, we can now finally state the result we have been after:

\begin{proposition}\label{proposition: existence of rccs}
Let $(\Omega, p)$ be a classical probability space with $\sigma$-algebra of random events $\Omega$ and probability measure $p$. For any $A,B\in \Omega$ satisfying (\ref{equation: positive correlation}) and any $n\geq 2$, there is some extension $(\Omega', p')$ of $(\Omega, p)$ whereby a \rccs of size $n$ exists for $(A, B)$.
\end{proposition}

\begin{proof}
The proof here is in all respects similar to Step 2 in \cite{hofer-szabo redei 2006}, the sole difference being that sets of admissible numbers should be replaced with sets of admissible* ones. This, in particular, will require replacing Hofer-Szab\'{o} and R\'{e}dei's equations (60)--(63) with:
\begin{gather}
r_{i}^{1} = \frac{c_{i}d_{i}}{\prob{A\wedge B}}										\label{equation: r1i} 	\\
r_{i}^{2} = \frac{c_{i}a_{i}-c_{i}d_{i}}{\prob{A\wedge \overline{B}}} 							\label{equation: r2i}	\\
r_{i}^{3} = \frac{c_{i}b_{i}-c_{i}d_{i}}{\prob{\overline{A}\wedge B}} 							\label{equation: r3i}	\\
r_{i}^{4} = \frac{c_{i} -c_{i}a_{i}-c_{i}b_{i}+c_{i}d_{i}}{\prob{\overline{A}\wedge \overline{B}}}.		\label{equation: r4i}	
\end{gather}
Owing to (\ref{equation: adm di*}), however,  Hofer-Szab\'{o} and R\'{e}dei's equations can be immediately recovered.
\end{proof}

Remarkably, Proposition \ref{proposition: existence of rccs} applies to \emph{every} pair of positively correlated events, be they strictly correlated or not. The analogous result announced by Hofer-Szab\'{o} and R\'{e}dei, on the contrary, concerned non-strictly correlated pairs only. The proof presented in this section, therefore, improves on their demonstration in two ways: it firstly rectifies the proof, but it also generalises the proof.


\section{Conclusion}


Reichenbachian common cause systems have been developed as a generalisation of the conjunctive fork model, to account for cases whereby the observed correlation between two causally independent events cannot be explained through the action of a single common cause. The existence of such systems in suitable extensions of the assumed probability space was allegedly demonstrated by Hofer-Szab\'{o} and R\'{e}dei. This paper has shown that their proof is logically deficient: the admissibility condition (\ref{equation: adm ab}) is necessary, but not sufficient, for the screening-off condition (\ref{rccs 1}). Accordingly, we propose an alternative admissibility condition (\ref{equation: adm di*}), which provides a more straightforward representation of the screening-off condition, and we demonstrate that the resulting set of admissibility conditions overcomes their logical error.



\end{document}